\documentclass[a4paper,UKenglish,cleveref, autoref]{lipics-v2019}

\usepackage{datetime}
\usepackage{breqn}
\usepackage{microtype}

\newtheorem{condit}{Conditions}

\title{Electronic markets with multiple submodular buyers}

\author{Allan Borodin}{Department of Computer Science, University of Toronto \and \url{http://www.cs.toronto.edu/~bor/}}{bor@cs.toronto.edu}{}{}
\author{Akash Rakheja}{Department of Computer Science, University of Toronto}{rakheja@cs.toronto.edu}{}{}

\authorrunning{A. Borodin and A. Rakheja} 

\Copyright{Allan Borodin and Akash Rakheja}
\ccsdesc[100]{Theory of computation~Market equilibria}
\keywords{Equilibria, market clearing, price of anarchy, posted price}

\category{}

\relatedversion{}

\supplement{}

\acknowledgements{}

\nolinenumbers 

\EventEditors{John Q. Open and Joan R. Access}
\EventNoEds{2}
\EventLongTitle{42nd Conference on Very Important Topics (CVIT 2016)}
\EventShortTitle{CVIT 2016}
\EventAcronym{CVIT}
\EventYear{2016}
\EventDate{December 24--27, 2016}
\EventLocation{Little Whinging, United Kingdom}
\EventLogo{}
\SeriesVolume{42}
\ArticleNo{23}

\usepackage{mathtools}
\usepackage{algorithm}
\usepackage{algpseudocode}

\newtheorem{observation}[theorem]{Observation}

\begin{document}

\maketitle

\begin{abstract}
We discuss the problem of setting prices in an electronic market that has more than one buyer. We assume that there are self-interested sellers each selling a distinct item that has an associated cost. Each buyer has a submodular valuation for purchasing any subset of items. The goal of the sellers is to set a price for their item such that their profit from possibly selling their item to the buyers is maximized. Our most comprehensive results concern a multi copy setting where each seller has $m$ copies of their item and there are $m$ buyers. In this setting, we give a necessary and sufficient condition for the existence of market clearing pure Nash equilibrium. We also show that not all equilibria are market clearing even when this condition is satisfied contrary to what was shown in the case of a single buyer in \cite{Babaioff2014}. 
Finally, we investigate the pricing problem for multiple buyers in the limited supply setting when each seller only has a single copy of their item. 
\end{abstract}


\section{Introduction}

How should sellers price their items in a market? This is arguably the central question in any market setting. The emergence of electronic markets has significantly changed the nature of markets. On the one hand, information collected online has made it possible for vendors to have detailed information on potential buyers. On the other hand, buyers now have easy access to the variety of items different sellers have available and their prices. Buyers are interested in purchasing a bundle of items maximizing their utility for the items purchased. Sellers are interested in maximizing their profit. Furthermore, the size and speed of online markets necessitates sellers to set their prices competitively and quickly.

Following Babaioff, Nisan and Paes Leme \cite{Babaioff2014}, we will view online markets as a full information game in which the sellers are the strategic agents. We are therefore interested in pricing equilibria in combinatorial markets. In a general setting of combinatorial markets with item pricing, we have many buyers, each having a publicly known valuation function $v$ for each subset of purchased items, and many strategic sellers, each having a set of items to sell at a price $p_i$ for item $i$. The utility of a buyer for subset $S$ is $v(S) - p(S)$ where $p(S) = \sum_{i \in S} p_i$. In general, the valuation function $v_j$ for a buyer $j$ can be quite complicated. The strategy of each seller is to set prices to maximize the prices for the items they sell. When items have a cost to produce, the seller wants to maximize their prices minus the costs for items sold. In this generality, the existence and properties (e.g., uniqueness, market clearing or not) of possible pure Nash Equilibria (NE) will depend on the precise setting. 

The main focus of the Babaioff et al paper is for a single buyer and strategic sellers each having a single copy of a distinct item. They establish basic results for different classes of valuation functions. Namely, they consider (in order of decreasing generality), arbitrary, subadditive, XOS, submodular, and gross substitutes valuations\footnote{A subadditive valuation is such that $v(S \cup T) \leq v(S) + v(T)$ $\forall S, T$. A XOS valuation is such that $v(S) = max_{t \in I} \sum_{i \in S} w_{it}$ for $w_{it} \in \mathbb{R}^+$. A gross substitutes valuation is such that if at a price $\textbf{p}$, a set $S$ is purchased and at a price $\textbf{p}' \geq \textbf{p}$, a set $S'$ is purchased then $S \cap \{i | p_i = p'_i\} \subseteq S'$. We will not consider these valuation classes but include these definitions for completeness.}. In our setting, we will restrict attention to submodular valuations and additive valuations. Babaioff et al also discuss multiple buyers in a Bayesian setting and how their results can be extended in the case of costs. In a following paper, Lev, Oren and Boutilier \cite{Lev2015} consider the setting of a single monotone submodular buyer where now each seller has a number of different items for sale. Our contribution will be to further the work of Babaioff et al with regard to more than one submodular buyer in a full information setting. Our work generalizes the previous settings in that we have multiple buyers each with their own valuations. We also assume that sellers have costs to produce an item. We extend the setting of Lev et al to multiple buyers, but restrict to each seller only selling copies of a distinct item. Hence, overall, our setting is incomparable to theirs. 

\section{Related results}

Walrus \cite{walrusJ03} pioneered the study of markets and market equilibrium (for {\it divisible items} being traded) in the late 19th century. Market equilibria is often referred to as competitive equilibria or Walraisian equlibria. The adaption of markets to buyers (having money) and sellers (with divisible items) is called Fisher markets (attributed to Irving Fisher and also dating back to the late 19th century; see Brainard and Scarf \cite{BrainardS00}). With regard to {\it indivisible items}, Gul and Stracchetti \cite{GulS99} provide a fundamental result showing that the class of Gross Substitutes valuations (a strict subset of submodular valuations) is the largest class of valuations containing unit demand buyers that always have a Walraisian equilibrium. Amongst other conditions, such equilibria require that prices are set so that all items are sold (i.e., the prices are a market clearing Nash equilibrum).

The topic of pricing is too extensive to indicate all the related work. It is rather remarkable, however, that precise characterization of pricing equilibria is a relatively recent topic of interest. As already indicated , our work is most closely related to the comprehensive paper of Babaioff et al \cite{Babaioff2014} for a single buyer and the following extension by Lev et al \cite{Lev2015} to sellers who can have more than one item for sale. 
Borodin, Lev and Strangway \cite{Borodin2016} consider another extension of Babaioff et al. Namely, they return to a single item per seller and a single buyer setting but now impose a budget on the buyer. We will not consider budgets although clearly budgets are an important consideration. Budgets raise many additional considerations for buyers and the resulting seller prices. This was already evident in the single buyer case. In our appendix, we will just illustrate how even in a very restricted setting, budgets will complicate the analysis of possible equilibria for multiple additive buyers.

\section{Preliminaries and Summary of Results}

We assume that $N = [n]$ is the set of sellers. Here, we use $[l]$ to denote the set $\{i \in \mathbb{N} | 1 \leq i \leq l\}$. We initially assume that seller $i$ sells item $i$ and has $k \leq m$ copies of this item to possibly sell to $m$ different buyers. We also assume there is a cost $c_i$ to seller $i$ to produce each copy of their item. Each seller $i \in N$ sets a price $p_i$ which is same for all copies of item $i$.

We assume that $M = [m]$ is the set of buyers and each buyer $j \in M$ has a submodular valuation function $v_j \colon 2^{N} \to {\mathbb R}_+$ where ${\mathbb R}_+$ is the set of all non-negative real numbers. We assume that each buyer is interested in buying at most one copy of each item. A submodular function $v \colon 2^{N} \to {\mathbb R}_+$ is a function such that $v(S \cup \{x\}) - v(S) \geq v(T \cup \{x\}) - v(T)$, $\forall S \subset T$ and $x \notin T$. We also assume that $v_j$ is normalized ($v_j(\emptyset) = 0$) and monotone ($v_j(S) \leq v_j(T)$ $\forall S, T \subseteq N$ such that $S \subseteq T$). For a given pricing $\textbf{p} = \{p_i\}_{i \in N}$, the utility of a buyer $j \in M$ for a subset $S \subseteq N$ of items is given by $u_j(S) = v_j(S) - p(S)$ where $p(S) = \sum_{i \in S} p_i$. Each buyer $j \in M$ chooses a set $S_{(j, \textbf{p})}$ such that $u_j(S_{(j, \textbf{p})}) = \max_S \{u_j(S) | S \subseteq N\}$. There can be more than one set that maximizes the utility, in which case a buyer arbitrarily (but deterministically) chooses a set that has the highest cardinality.
Thus every pricing $\textbf{p}$ induces a given allocation of items 
$S_{j,\textbf{p}}$ to each buyer $j \in M$. 
The social welfare $SW(\textbf{p})$ of an allocation induced by a pricing is defined as the sum $\sum_{j \in M} (v_j(S_{j,\textbf{p}}) - c(S_{j, \textbf{p}}))$.

There are different pricing games that can be considered when there are many buyers and sellers. In the mechanism design area, we usually study combinatorial auctions for a single seller and multiple buyers who are the strategic agents. In contrast, in the pricing game as initiated in Babaioff et al, there is one utility maximizing buyer and it is the sellers who are the strategic agents whose actions consist solely of setting prices. But once there is a limited supply (i.e., more buyers than copies) of each item, it is not clear which buyers will obtain a given item. Since we do not view the buyers as being strategic (i.e., their valuations are true and known), either the sellers or an outside mechanism have to determine an order in which to serve buyers. That is, we are essentially considering an online algorithm in which buyers arrive sequentially. In the limited supply case, we will study such an online setting and 
leave other extensions for future work.


A pure Nash Equilibrium (PNE) of the game is given by a pricing $\textbf{p} \in (\mathbb{R}_+)^n$. Under this pricing $\textbf{p}$, no seller can unilaterally increase profit by changing their price. This is an equilibrium in the game between sellers. The profit of a seller $i \in N$ is given by $\alpha_i(p_i - c_i)$ where $\alpha_i$ denotes the number of buyers that buy item $i$ and $c_i$ is the cost of production to seller $i$. The cost of production to one or more sellers may be 0. A market clearing Pure Nash Equilibrium is a PNE such that all copies of all the items are sold by the sellers.
In addition to considering the standard Nash equilibrium, we will sometimes consider $\epsilon$-Nash equilibria for sufficiently small $\epsilon > 0$ in which any individual seller cannot gain more than an additive profit of $\epsilon$ by deviating.

Let $OPT = \max_{\textbf{p}} SW(\textbf{p})$. The {\it price of anarchy} (POA) for a market game instance $I$ is then defined as $\max_{\{\textbf{p}: \textbf{p} \text{ is a PNE}\}} \frac{OPT}{SW(\textbf{p})}$. 
That is, the POA is the worst case approximation induced by a PNE. The {\it price of stability} (POS) is defined as $\min_{\{\textbf{p}: \textbf{p} \text{ is a PNE}\}} \frac{OPT}{SW(\textbf{p})}$. That is, the best case approximation induced by a PNE.

As stated, throughout this paper we assume that the buyer valuation functions $v_j$ are monotone submodular for each $j \in M$. For some results, we further assume that the $v_j$ are additive for one or more buyers. We define $v_j(i | S)$ as $v_j(S \cup \{i\}) - v_j(S)$ for $S\subseteq N$, $i \in N \setminus S$, $j \in M$. We denote by $v_{j,i}$ the marginal value $v_j(i | N \setminus \{i\}) = v_j(N) - v_j(N \setminus \{i\})$ $\forall i \in N$ $\forall j \in M$. We note that $v_j(i | S) \geq v_{j, i}$, $\forall j \in M$ and $\forall S \subseteq N$ $\forall i \in N \setminus S$ by the definition of submodularity.

\subsection{Summary of results}

In sections \ref{sec:multiple-copies} and \ref{sec:many-buyers} (and also Appendix \ref{sec:budget}), we have as many copies of items as we have buyers. Hence, buyers are never in competition for an item. In sections \ref{sec:multiple-copies} and \ref{sec:many-buyers} we provide a rather comprehensive set of results including a characterization for when there is a unique market clearing pure equilibirium and a sufficient condition for when there is a unique equilibrium. 
In the unlimited supply setting with no budgets, we consider the social welfare obtained in equilibria, and the resulting price of anarchy and price of stability.
Appendix \ref{sec:budget} considers a very restricted setting in which a budget is introduced. In this setting, we do not have explicit prices but rather we have a set of conditions that can easily be tested and these conditions are both necessary and sufficient for a market clearing equilibrium.

In section \ref{sec:limited-supply}, we have a limited supply of items.
As stated, we need to have some means to resolve conflicts amongst the buyers. In particular, we will assume that each seller has only one copy of their item. We begin with some general observations that hold no matter how items are allocated to buyers. In section \ref{subsec:one-submodular-buyer}, we consider the restricted setting where there is one submodular buyer and the remaining buyers are additive. We assume that the sellers (or a central authority) have agreed upon a fixed order in which agents arrive and get their choice of items. The negative results (i.e., Examples \ref{Example:4.1}, \ref{Example:4}) therefore hold in the online setting where an adversary sets the order of arrivals. 
This is similar to online posted price mechanisms as studied for the case of a single seller. Here, however, we 
have more than one seller and for any positive claims about equilibria, we will assume that buyers arrive in a fixed order that is known to all sellers.

Our positive results for the limited supply setting as with regard to the existence and characterization of pure Nash equilibira and market clearing equilibria are restricted to some very special settings.
Our negative examples suggest that any extension of our positive results to more general markets (with multiple submodular buyers) will not follow in any immediate way. 

As an indication of the difference between previous results for a single buyer discussed in \cite{Babaioff2014} where there is always a pure Nash equilibrium in the case of submodular valuations, we provide an example where no $\epsilon$-Nash Equilibirum exists in the limited supply setting. 
In the following example there is no $\epsilon$-Nash equilibrium for a sufficiently small $\epsilon$.

\begin{example}
\label{Example:4.1}
Consider the market where $M = [2]$, $N = [2]$ and the valuation function is given by the following table:
\begin{center}
\begin{tabular}{ c| c c c }
 sets & 1 & 2 & 1, 2 \\
 \hline
 $v_1$ & 16 & 16 & 32 \\
 $v_2$ & 30 & 30 & 45 \\
\end{tabular}
\end{center}
$v_1$ is additive and $v_2$ is submodular as can be easily verified. There is no pure $\epsilon$-Nash equilibrium for $\epsilon < 7$ when buyer $1$ is preferred by every seller over buyer $2$ at a given pricing when both the buyers are willing to purchase. At any given pricing, buyer $1$ chooses a set $S_1 \subseteq N$ that it consumes, and then buyer $2$ chooses a subset of $N \setminus S_1$ that maximizes its utility.
\end{example}

\begin{proof}
Let $\epsilon < 7$ and suppose there is an $\epsilon$-Nash Equilibrium $\textbf{p}$ and under pricing $\textbf{p}$, $S_1$ is sold to $1$ and $N \setminus S_1$ to $2$. The equilibrium has to be market clearing because each seller can earn a profit of at least $\epsilon$ by selling to buyer 1 at a price of 16.

Let $S_1 = N$ and $S_2 = \emptyset$. Then, $p_1, p_2 \leq 16$ since otherwise buyer 1 will not purchase those items. Now, if seller $2$ increases the price to 30, buyer 1 will purchase set $\{1\}$ and buyer 2 will purchase $\{2\}$ increasing the profit of seller 2 by more than $\epsilon$. So, there is no equilibrium with $S_1 = N$, $S_2 = \emptyset$.

Let $S_1 = \{1\}$ and $S_2 = \{2\}$. Then, $p_1 \leq 16$ since otherwise buyer 1 will not purchase item $1$ and $16 < p_2 \leq 30$ since buyer 1 does not purchase but buyer 2 does. If $p_2 > 23$, seller $1$ can increase its price to 23 and be sold to buyer 2 increasing the profit by more than $\epsilon$. If $p_2 \leq 23$, seller $2$ can increase the price to 30 and still be sold to buyer 2 increasing the profit by more than $\epsilon$. A similar thing happens for the case of $S_1 = \{2\}$, $S_2 = \{1\}$. So, no equilibria exists in these cases.

Let $S_1 = \emptyset$ and $S_2 = N$. Then. $p_1, p_2 \leq 15$ since otherwise the buyer 2 will purchase at most one item. At this price, buyer 1 will already purchase both the items and hence this case is not possible.

Thus, an $\epsilon$-Nash equilibrium does not exist in this case.
\end{proof}




\section{Characterization of Market Clearing Equilibrium in case of two copies and two submodular buyers}
\label{sec:multiple-copies}

We begin by considering the case of two buyers and two copies of each item, i.e., $M = \{1, 2\}$. In this case, the buyer order does not matter. Also, here {\it market clearing} means that each buyer buys from each seller. We establish the necessary and sufficient conditions for the existence of market clearing pure Nash equilibrium as stated in the following theorem, and proved using observation \ref{Observation:ci}, lemma \ref{Lemma:1} and lemma \ref{Lemma:2}. The proof of this theorem along with other proofs in this section are subsumed in section \ref{sec:many-buyers} but for notational convenience, we state them separately.

\begin{theorem}
\label{Theorem:twobuyers}
The necessary and sufficient condition in the case of 2 buyers ($M = [2] = \{1, 2\}$) and n sellers ($N = [n]$) for the existence of a market clearing pure Nash equilibrium is that $\forall i \in N$ $\frac{v_{2,i} - c_i}{2} \leq v_{1, i} - c_i \leq 2(v_{2,i} - c_i)$ and $\min \{v_{1, i}, v_{2, i}\} \geq c_i$. Moreover, when this condition is satisfied, the market clearing pure Nash equilibrium is unique.
\end{theorem}

We contrast this theorem by giving an example where the condition is not satisfied and no $\epsilon$-Nash equilibrium exists.

\begin{example}
\label{Example:noequi}
Consider the market where $M = [2]$, $N = [2]$, the cost to sellers is $0$ and the valuation function is given by the following table:
\begin{center}
\begin{tabular}{ c| c c c c c c c }
 sets & 1 & 2 & 1, 2 \\
 \hline
 $v_1$ & 200 & 210 & 220 \\
 $v_2$ & 70 & 50 & 120 \\
\end{tabular}
\end{center}
Valuation $v_1$ is submodular and $v_2$ is additive as can be easily verified. We will show that for $\epsilon < 5$, there cannot be a pure $\epsilon$-Nash equilibrium for these valuations. 
\end{example}

\begin{proof}
Let $\textbf{p}$ be an $\epsilon$-Nash equilibrium for some $\epsilon < 5$. Let $S_1$ be the set purchased by buyer 1 and $S_2$ be the set purchased by buyer 2 at this pricing. There are four possibilities for each $S_1$ and $S_2$ namely $N, \{1\}, \{2\}$ and $\emptyset$. We will consider all possibilities and in each case show that some seller can improve upon their revenue by deviating. We do not need to consider the combinations where $S_1 \cup S_2 \neq N$ since every seller can sell to buyer 2 by setting a price $50$ earning strictly more profit than $\epsilon$ and hence $S_1 \cup S_2 = N$. 

\begin{itemize}
\item
If $S_1 = S_2 = N$, then $p_1 \leq 10$ since otherwise buyer 1 will prefer $\{2\}$ over $N$ and hence seller 1 can profit by more than $\epsilon$ by setting a price of $70$ and just selling to buyer 2.

\item 
If $S_1 = \{1\}$ and $S_2 = N$, then $p_2 \leq 50$ and $p_1 \leq p_2 - 10$ since otherwise buyer 1 will prefer $\{2\}$ over $\{1\}$. If $p_1 > 25$ then by setting a price of $p_1 + 9$, seller 2 will earn a profit of $2p_1 + 18 > 68$ increasing it by more than $\epsilon$. If $p_1 \leq 25$, seller 1 can profit by more than $\epsilon$ by setting a price of $70$ and just selling to buyer 2.

\item 
The case of $S_1 = \emptyset$ and $S_2 = N$ is not a possible equilibrium since in this case $p_1 \leq 70$ and hence buyer 1 will prefer $\{1\}$ over $\emptyset$.

\item 
If $S_1 = \{2\}$ and $S_2 = N$, then $p_1 \leq 70$, $p_2 \leq 50$ and $p_1 \geq p_2 - 10$ since otherwise buyer 1 will prefer $\{1\}$ over $\{2\}$. If $p_1 \leq 60$ then by setting a price of $70$, seller 1 will increase his profit by more than $\epsilon$. If $p_2 \leq 47.5$ and $p_1 > 60$, seller 2 can profit by more than $\epsilon$ by setting a price of $50$. If $p_1 > 60$ and $p_2 > 47.5$, then by setting a price of $37.5$ seller 1 can sell to both buyers earning a profit of $75$ increasing it by at least $\epsilon$.

\item 
The case of $S_1 = N$ and $S_2 \neq N$ is not possible since in this case $p_1 \leq 10$ and $p_2 \leq 20$, so buyer 2 will purchase $N$.

\item 
If $S_1 = \{2\}$ and $S_2 = \{1\}$, then $p_1 \leq 70$ and $p_2 > 50$. By setting a price of 40, seller 1 can sell to both the buyers and earn a profit of $80$ increasing the profit by at least $\epsilon$.

\item 
If $S_1 = \{1\}$ and $S_2 = \{2\}$, then $p_1 > 70$ and $p_2 \leq 50$. This case is not possible since at this pricing, buyer 1 will prefer $\{2\}$ over $\{1\}$.

\end{itemize}

So, $\textbf{p}$ is not an $\epsilon$-Nash equilibrium for $\epsilon < 5$. Since $\textbf{p}$ is arbitrary, there is no $\epsilon$-Nash equilibrium for $\epsilon < 5$.
\end{proof}



We first prove that the given condition is a necessary condition. We start by making three observations that will be helpful in the proof.

\begin{observation}
\label{Observation:lp}
For a buyer with submodular valuation function $v$, if some item $i$ has a price of at most $v_i = v(i|N \setminus \{i\})$, the buyer will prefer to purchase item $i$.
\end{observation}

\begin{proof}
Let us consider a pricing $\textbf{p}$ such that $p_i \leq v_{i}$. Let us assume that the buyer will buy a set $S$ of items that does not contain $i$ under this pricing $\textbf{p}$. Then, $v(i | S) = v(S \cup \{i\}) - v(S) \geq v_{i} \geq p_i$ (since $v$ is submodular) and hence $v(S \cup \{i\}) - p(S \cup \{i\}) \geq v(S) - p(S)$, so under pricing $\textbf{p}$, the buyer will prefer to purchase $i$ due to the tie-breaking rule which breaks ties in favor of larger sets. 
\end{proof}

\begin{observation}
\label{Observation:prefbuy}
For a buyer with submodular valuation function $v$, if some item $i$ has a price of at most $v(i|S)$ for some $S \subseteq N$, $i \notin S$, the buyer will prefer to purchase $S \cup \{i\}$ over $S$.
\end{observation}

\begin{proof}
Let us consider a pricing $\textbf{p}$ such that $p_i \leq v(i|S)$. Then, $v(i | S) = v(S \cup \{i\}) - v(S) \geq p_i$ and hence $v(S \cup \{i\}) - p(S \cup \{i\}) \geq v(S) - p(S)$, so under pricing $\textbf{p}$, the buyer will prefer to purchase $S \cup \{i\}$ over $S$ due to the tie-breaking rule which breaks ties in favor of larger sets. 
\end{proof}


\begin{observation}
\label{Observation:ci}
Let $N = [n]$, $M = [2]$ and suppose there exists $i \in N$ such that $\min \{v_{1, i}, v_{2, i}\} < c_i$. Then, there is no market clearing pure Nash equilibrium.
\end{observation}

\begin{proof}
Let us assume, for the sake of contradiction, that there is a market clearing pure Nash equilibrium given by pricing $\textbf{p}$. $p_i \leq v_{2, i}$ since if $p_i > v_{2, i} = v_2(N) - v_2(N \setminus \{i\})$, then $v_2(N \setminus \{i\}) - p(N \setminus \{i\}) > v_2(N) - p(N)$ since $p_i = p(N) - p(N \setminus \{i\})$ and hence buyer 2 will purchase $N \setminus \{i\}$ instead of $N$. Similarly, $p_i \leq v_{1, i}$. Thus, profit of seller $i$ is $2(p_i - c_i) \leq 2(\min \{v_{1, i}, v_{2, i}\} - c_i) < 0$ and hence seller $i$ will earn more profit by deviating to price $c_i$. Since seller $i$ has an incentive to deviate, $\textbf{p}$ is not an equilibrium.

Since $\textbf{p}$ was assumed to be an arbitrary market clearing pure Nash equilibrium, there is none.
\end{proof}

\begin{lemma}
\label{Lemma:1}
Let $N = [n]$, $M = [2]$ and suppose there exists $i \in N$ such that $v_{1,i} - c_i > 2(v_{2, i} - c_i)$. Then, there is no market clearing pure Nash equilibrium.
\end{lemma}

\begin{proof}
Let us assume, for the sake of contradiction, that there is a market clearing pure Nash equilibrium given by pricing $\textbf{p}$. $p_i \leq v_{2, i}$ since if $p_i > v_{2, i} = v_2(N) - v_2(N \setminus \{i\})$, then $v_2(N \setminus \{i\}) - p(N \setminus \{i\}) > v_2(N) - p(N)$ since $p_i = p(N) - p(N \setminus \{i\})$ and hence buyer 2 will purchase $N \setminus \{i\}$ instead of $N$.


Since $\textbf{p}$ is market clearing, both buyers will purchase all items and hence the profit of seller $i$ will be $2(p_i - c_i) \leq 2(v_{2, i} - c_i) < v_{1, i} - c_i$. If seller $i$ instead sets price $p'_i = v_{1, i}$, buyer $1$ will still purchase item $i$ due to observation \ref{Observation:lp} and new profit of seller $i$ will be $v_{1, i} - c_i$. So, seller $i$ has an incentive to deviate from $\textbf{p}$ and hence $\textbf{p}$ is not an equilibrium.

Since $\textbf{p}$ was assumed to be an arbitrary market clearing pure Nash equilibrium, there is none.
\end{proof}

Lemma \ref{Lemma:1} proves that $\forall i \in N$ $v_{1, i} - c_i \leq 2(v_{2,i} - c_i)$ is a necessary condition for the existence of a market clearing pure Nash equilibrium in the case of 2 buyers. Similarly, $\forall i \in N$ $v_{2, i} - c_i \leq 2(v_{1, i} - c_i)$ is also a necessary condition. Thus, $\forall i \in N$ $\frac{v_{2,i} - c_i}{2} \leq v_{1, i} - c_i \leq 2(v_{2,i} - c_i)$ is a necessary condition for the existence of a market clearing pure Nash equilibrium in the case of 2 buyers. Also, from observation \ref{Observation:ci}, $\min \{v_{1, i}, v_{2, i}\} \geq c_i$ is a necessary condition. Now, we will prove that these two conditions together are also a sufficient condition.

\begin{lemma}
\label{Lemma:2}
There is a unique market clearing pure Nash equilibrium $\textbf{p}$ in the case of $M = [2]$ when $\forall i \in N$ $\frac{v_{2,i} - c_i}{2} \leq v_{1, i} - c_i \leq 2(v_{2,i} - c_i)$ and $\min \{v_{1, i}, v_{2, i}\} \geq c_i$. Namely, $p_i = \min \{v_{1, i}, v_{2, i}\}$ $\forall i \in N$.
\end{lemma}

\begin{proof}
Let $p_i = \min \{v_{1, i}, v_{2, i}\}$ $\forall i \in N$.


Since $\forall i \in N$ $p_i \leq v_{1, i}$ and $p_i \leq v_{2, i}$, it is clear that both the buyers will purchase all the items due to observation \ref{Observation:lp}. Thus, there is no profit in reducing the price for any seller. The profit in this case for seller $i$ is $2(p_i - c_i) \geq 0$.

Without loss of generality, we consider the case where for some $i \in N$ $v_{1, i} \geq v_{2, i}$. In this case, $p_i = v_{2, i}$. Let us assume that seller $i$ increases the price of his item to $p'_i > p_i$ and this new pricing is called $\textbf{p}'$. Let us also assume that buyer $2$ purchases a set $S$ containing $i$ under this new pricing. If $S \neq N$, then for all $j \in N \setminus S$, $v_2(j | S) = v_2(S \cup \{j\}) - v_2(S) \geq v_2(N) - v_2(N \setminus j) = v_{2, j} \geq p_j$, so $v_2(S \cup \{j\}) - p'(S \cup \{j\}) \geq v_2(S) - p'(S)$ since $p'_j = p_j$. Hence buyer $2$ will prefer $S \cup \{j\}$ over $S$ due to tie-breaking rule. This is a contradiction and hence $S = N$. Since $p'_i > v_{2, i} = v_2(N) - v_2(N \setminus \{i\})$, $v_2(N \setminus \{i\}) - p'(N \setminus \{i\}) > v_2(N) - p'(N)$, so buyer 2 will not purchase $i$. Similarly if seller $i$ increases the price to more than $v_{1,i}$, buyer $1$ will not purchase item $i$. So, the maximum profit that seller $i$ can earn by increasing the price is $v_{1, i} - c_i$. Since $v_{1, i} - c_i \leq 2(v_{2, i} - c_i)$, there is no incentive to increase the price. Therefore, $\textbf{p}$ is a pure Nash equilibrium.

Let us assume that there is another market clearing pure Nash Equilibrium $\textbf{p}'$ such that $p'_i \neq \min \{v_{1, i}, v_{2, i}\}$ for some $i \in N$. If $p'_i > v_{1, i}$, then buyer $1$ will prefer to buy set $N \setminus \{i\}$ over set $N$ since $v_1(N) - v_1(N \setminus \{i\}) = v_{1, i} < p'_i = p'(N) - p'(N \setminus \{i\})$ and hence $v_1(N) - p'(N) < v_1(N \setminus \{i\}) - p'(N \setminus \{i\})$. Therefore, $p'_i \leq v_{1, i}$. Similarly, $p'_i \leq v_{2, i}$. If $p'_i < \min \{v_{1, i}, v_{2, i}\}$, then by setting price $\min \{v_{1, i}, v_{2, i}\}$, seller $i$ will still sell both the copies of his item due to observation \ref{Observation:lp} and hence will earn strictly more profit. Thus, $\textbf{p}'$ is not a market clearing pure Nash Equilibrium if $p'_i \neq \min \{v_{1, i}, v_{2, i}\}$ for some $i \in N$.

Hence, $\textbf{p}$ is the unique market clearing pure Nash equilibrium.
\end{proof}

When $c_i > \min \{v_{1, i}, v_{2, i}\}$ for some $i$, there are no market clearing pure Nash equilibria as stated in Observation \ref{Observation:ci}.
Market clearing pure Nash equilibria are also total welfare maximizing in the case that $\forall i \in N$ $c_i \leq \min \{v_{1, i}, v_{2, i}\}$ since total welfare is the sum of the utility of all buyers and sellers. When buyer 1 buys a set $S_1$ and buyer 2 buys a set $S_2$, the money is transferred from buyers to sellers and hence the total welfare is $(v_1(S_1) - p(S_1)) + (v_2(S_2) - p(S_2)) + (p(S_1) - c(S_1)) + (p(S_2) - c(S_2)) = v_1(S_1) - c(S_1) + v_2(S_2) - c(S_2)$ which is maximum when $S_1 = S_2 = N$, that is when the market is cleared. 
Thus, when market clearing pure Nash equilibria exists, the price of stability is 1.

For the case of a single submodular buyer, and each seller having a single copy of their distinct item, Babaioff et al \cite{Babaioff2014} show that there is a unique pure Nash equilibrium (NE) which is market clearing and that no other pure NE can exist. That is, the one and only pure Nash is market clearing. We now show that in the setting of two buyers and each seller having two copies of their item, there can be other Nash equilibrium that are not market clearing. These non-market clearing NE can exist even when the necessary and sufficient conditions hold for the unique market clearing NE. Consider the case of two additive buyers, where $v_{1, i} - c_i = 2(v_{2, i} - c_i) \geq 0$ for some seller $i \in N$, the seller $i$ is indifferent between setting price $v_{1, i}$ and $v_{2, i}$ since in both cases, the seller will earn equal profit. So, the pure Nash equilibrium in which it sets a price of $v_{1, i}$ is not market clearing since buyer 2 will not purchase at this price. So, in case $v_{1, i} - c_i = 2(v_{2, i} - c_i)$ or $2(v_{1, i} - c_i) = v_{2, i} - c_i$ for some $i \in N$, there are pure Nash Equilibria which are not market clearing. Also in the case of two additive buyers, if $v_{1, i} = v_{2, i} = c_i$, seller $i$ will always earn non-positive profit and hence can set any arbitrarily high price, and raising or lowering that price does not increase the profit of seller $i$. So in this case as well, there can be many pure Nash Equilibria which are not market clearing. Furthermore, even if the inequalities in the characterization of market clearing equilibria are strict inequalities, there can still be pure Nash equilibria that are not market clearing. 
That is, even when $\forall i \in N$ $\frac{v_{2,i} - c_i}{2} < v_{1, i} - c_i < 2(v_{2,i} - c_i)$ and $\min \{v_{1, i}, v_{2, i}\} > c_i$, the uniqueness of pure Nash equilibrium is not guaranteed. In particular, there can sometimes be pure Nash Equilibrium that are non-market clearing when both the buyers are submodular.

\begin{example}
\label{Example:2}
Consider the market where $M = [2]$, $N = [3]$, the cost to sellers is $0$ and the valuation function is given by the following table:
\begin{center}
\begin{tabular}{ c| c c c c c c c }
 sets & 1 & 2 & 3 & 1, 2 & 1, 3 & 2, 3 & 1, 2, 3 \\
 \hline
 $v_1$ & 90 & 90 & 80 & 100 & 100 & 101 & 110 \\
 $v_2$ & 70 & 80 & 85 & 86 & 90 & 100 & 105 \\
\end{tabular}
\end{center}
Both the valuation functions are submodular as can be easily verified. The pricing $\textbf{p} = (10, 10, 20)$ with $S_1 = \{1, 2\}$ and $S_2 = \{2, 3\}$ is a pure Nash equilibrium which is not market clearing.
\end{example}

In contrast to the previous example, we can show the uniqueness of pure Nash equilibrium when at least one of the valuation functions is additive. Without loss of generality, we assume that $v_1$ is submodular and $v_2$ is additive. 

\begin{theorem}
\label{Theorem:twobuyersunique}
When $N = [n]$, $M = [2]$, $v_1$ is submodular, $v_2$ is additive, and $\forall i \in N$ $\frac{v_{2,i} - c_i}{2} < v_{1, i} - c_i < 2(v_{2,i} - c_i)$ and $\min \{v_{1, i}, v_{2, i}\} > c_i$ then there is only one pure Nash Equilibrium and that equilibrium is market clearing. This equilibrium is given by $p_i = \min \{v_{1, i}, v_{2, i}\}$ $\forall i \in N$.
\end{theorem}

\begin{proof}
Let $\textbf{p}$ be an arbitrary pure Nash equilibrium. For each $i \in N$, $p_i \geq \min \{v_{1, i}, v_{2, i}\}$ since if $p_i$ is smaller than both $v_{1, i}$ and $v_{2, i}$, seller $i$ can increase the price and still sell to both the buyers as proved earlier in observation \ref{Observation:lp}. This will only increase his profit.

Let $S_1$ and $S_2$ be the sets purchased by the two buyers respectively. Let $i \in S_2 \setminus S_1$ be arbitrary. Then, $p_i > v_1(i | S_1)$ since otherwise buyer 1 will prefer $S_1 \cup \{i\}$ over $S_1$ due to observation \ref{Observation:prefbuy}. Also, $p_i = v_{2, i} > v_1(i | S_1) \geq v_{1, i}$ since $v_2$ is additive and hence buyer 2 will purchase at maximum price $v_{2, i}$. Since $v_{2, i} - c_i < 2(v_{1, i} - c_i) \leq 2(v_1(i | S_1) - c_i)$, it is better for seller $i$ to set the price equal to $p'_i = v_1(i | S_1)$ since at this price, $v_1(S_1 \cup \{i\}) - p'(S_1 \cup \{i\}) \geq v_1(S_1) - p'(S_1)$ due to submodularity and hence in this case both the buyers will purchase item $i$ due to the tie-breaking rule and seller $i$ will earn more profit. Thus, the set $S_2 \setminus S_1$ is empty.

If for some seller $i$, $i \notin S_1 \cup S_2$, then by setting price $p_i = \min \{v_{1, i}, v_{2, i}\}$, the seller will sell to both buyers and earn a positive profit. So, $S_1 \cup S_2 = N$ in equilibrium. Therefore, $S_1 = N$.

Let some $i \in S_1 \setminus S_2$. In this case, $p_i \geq v_{1, i}$ since otherwise seller $i$ can earn more profit by increasing the price. Also, $p_i \leq v_{1, i}$ since otherwise $N \setminus \{i\}$ is more valuable to buyer 1 than $S_1 = N$. Therefore, $p_i = v_{1, i}$. In this case, $v_{1, i} > v_{2, i}$ and setting price $v_{2, i}$ is more profitable since by doing so, the profit increases to $2(v_{2, i} - c_i) > v_{1, i} - c_i$. Thus, the set $S_1 \setminus S_2$ is empty.

Therefore, $S_1 = S_2 = N$. Therefore, $p_i \leq v_{2, i}$ and $p_i \leq v_{1, i}$ since otherwise at least one of the buyers will not purchase $i$.
Hence, $p_i = \min \{v_{1, i}, v_{2, i}\}$ gives a unique pure Nash Equilibrium and this equilibrium is market clearing. 
\end{proof}

Note that the price can be as high as $v_2(i | S_2 \setminus \{i\})$ in case $v_2$ is an arbitrary submodular valuation and it might hold that $v_2(i | S_2 \setminus \{i\}) - c_i > 2(v_1(i | S_1) - c_i)$ in which case selling to both buyers is not profitable for seller $i$ and hence it is not necessary that $S_2 \setminus S_1$ is empty. Thus, the proof will not follow.

The following corollary follows from the above theorem.

\begin{corollary}
\label{Corollary:uniquePoA}
When $N = [n]$, $M = [2]$, $v_1$ is submodular, $v_2$ is additive, and $\forall i \in N$ $\frac{v_{2,i} - c_i}{2} < v_{1, i} - c_i < 2(v_{2,i} - c_i)$ and $\min \{v_{1, i}, v_{2, i}\} > c_i$, the price of anarchy is 1.
\end{corollary}

\section{The Unlimted Supply Setting for an Arbitrary Number of Buyers}
\label{sec:many-buyers}

Generalising the case of two buyers, the necessary and sufficient condition in case of $m$ buyers and $m$ copies per seller (i.e., $M = [m]$) is given by the following theorem. We prove it using observation \ref{Observation:ci2}, lemma \ref{Lemma:12} and lemma \ref{Lemma:13}. Recall that $v_{j, i} = v_j(N) - v_j(i | N \setminus \{i\})$ $\forall j \in M$ $\forall i \in N$.

For all $i \in N$, we define $\pi_i$ to be a permutation of $[m]$ in non-decreasing order of $v_{\pi_i(j), i}$ ($v_{\pi_i(k), i} \leq v_{\pi_i(k+1), i}$ $\forall k \in [m-1]$ and hence $v_{\pi_i(1),i} = \min_j v_{j,i}$).

\begin{theorem}
The necessary and sufficient condition in case of m buyers ($M = [m]$) and n sellers ($N = [n]$) for the existence of a market clearing pure Nash equilibrium is $\forall i \in N$ $\forall j \in [m]$ $v_{\pi_i(j), i} - c_i \leq \frac{m}{m-j+1}(v_{\pi_i(1), i} - c_i)$ and $\forall i \in N$ $\min_{j \in [m]} v_{\pi_i(j), i} \geq c_i$.
\end{theorem}

The necessary condition follows from the following observation and lemma.

\begin{observation}
\label{Observation:ci2}
Let $N = [n]$, $M = [m]$ and suppose there exists $i \in N$ such that $\min_{j \in [m]} v_{\pi_i(j), i} < c_i$. Then, there is no market clearing pure Nash equilibrium.
\end{observation}

\begin{proof}
Let us assume, for the sake of contradiction, that there is a market clearing pure Nash equilibrium given by pricing $\textbf{p}$. $p_i \leq v_{\pi_i(j), i}$ since if $p_i > v_{\pi_i(j), i} = v_{\pi_i(j)}(N) - v_{\pi_i(j)}(N \setminus \{i\})$, then $v_{\pi_i(j)}(N \setminus \{i\}) - p(N \setminus \{i\}) > v_{\pi_i(j)}(N) - p(N)$ since $p_i = p(N) - p(N \setminus \{i\})$ and hence buyer $\pi_i(j)$ will purchase $N \setminus \{i\}$ instead of $N$. Thus, $p_i \leq \min_{j \in [m]}v_{\pi_i(j), i}$. Thus, profit of seller $i$ is $m(p_i - c_i) \leq m(\min_{j \in [m]}v_{\pi_i(j), i} - c_i) < 0$ and hence seller $i$ will earn more profit by deviating to price $c_i$. Since seller $i$ has an incentive to deviate, $\textbf{p}$ is not an equilibrium.

Since $\textbf{p}$ was assumed to be an arbitrary market clearing pure Nash equilibrium, there is none.
\end{proof}

\begin{lemma}
\label{Lemma:12}
If for some $i \in N$ and some $j \in [m]$, $v_{\pi_i(j), i} - c_i > \frac{m}{m-j+1}(v_{\pi_i(1), i} - c_i)$, then there is no market clearing pure Nash Equilibrium.
\end{lemma}

\begin{proof}
Let us assume to the contrary that there is a market clearing pure Nash Equilibrium $\textbf{p}$. Then, $\forall k \in N$ $p_k \leq v_{\pi_k(l), k}$ $\forall l \in [m]$ since otherwise some buyer $\pi_k(l)$ will not purchase item $k$ since the utility from $N \setminus \{k\}$ will be strictly higher than from $N$. Also, $\forall k \in N$ $p_k \geq \min_{l \in [m]}{v_{\pi_k(l), k}}$ since otherwise seller $k$ can increase price and still sell to all the buyers as shown in observation \ref{Observation:lp}. Therefore, $p_k = \min_{l \in [m]}{v_{\pi_k(l), k}} = v_{\pi_k(1), k}$ $\forall k \in N$. Thus, the profit of seller $k$ under this pricing is $m(p_k - c_k) = m(v_{\pi_k(1), k} - c_i)$.

Now, if seller $i$ deviates to a price of $v_{\pi_i(j), i}$, it will still sell to $m-j+1$ buyers since new price $v_{\pi_i(j), i} \leq v_{\pi_i(l), i}$ $\forall l \geq j$. Thus, the new profit of seller $i$ is $(m-j+1)(v_{\pi_i(j), i} - c_i) > m(v_{\pi_i(1), i} - c_i)$ and hence the seller $i$ has an incentive to deviate. Therefore, $\textbf{p}$ is not a market clearing pure Nash Equilibrium.

Since $\textbf{p}$ was assumed to be an arbitrary market clearing pure Nash Equilibrium, there is none.
\end{proof}

The sufficient condition follows from the following lemma.

\begin{lemma}
\label{Lemma:13}
If $\forall i \in N$ $\forall j \in [m]$ $v_{\pi_i(j), i} - c_i \leq \frac{m}{m-j+1}(v_{\pi_i(1), i} - c_i)$ and $\forall i \in N$ $\min_{j \in [m]} v_{\pi_i(j), i} \geq c_i$, then there is a unique market clearing pure Nash Equilibrium.
\end{lemma}

\begin{proof}
Let $\textbf{p}$ be a pricing such that $p_i = v_{\pi_i(1), i} = \min_{k \in M} v_{k,i}$ $\forall i \in N$.

At this pricing, all the items are sold to all the buyers as shown in observation \ref{Observation:lp} since $p_i \leq v_{j, i}$ $\forall j \in M$ $\forall i \in N$. Thus, no seller has an incentive to lower the price. The profit to seller $i$ at this pricing is $m(v_{\pi_i(1), i} - c_i) \geq 0$.

Let us consider that some seller $i$ increases its price to $p'_i$ and let the new pricing be $\textbf{p}'$. Since $p'_k \leq v_{j, k}$ $\forall j \in M$ $\forall k \in N \setminus \{i\}$, every buyer will purchase all the items other than $i$. Let $j$ be the minimum index such that $p'_i \leq v_{\pi_i(j), i}$. If there is no such $j$, then no buyer will purchase item $i$ and hence this deviation is not profitable to seller $i$. Item $i$ will be purchased by $m-j+1$ sellers and hence the maximum profit seller $i$ will earn is $(m-j+1)(p'_i - c_i) \leq (m-j+1)(v_{\pi_i(j), i} - c_i) \leq m(v_{\pi_i(1), i} - c_i) = m(p_i - c_i)$ and hence it is not profitable for seller $i$ to deviate.

Thus, $\textbf{p}$ is a market clearing pure Nash Equilibrium.

Let us assume that there is another market clearing pure Nash Equilibrium $\textbf{p}'$. Since it is market clearing, $p'_i \leq v_{\pi_i(1), i}$ $\forall i \in N$. If $p'_i < v_{\pi_i(1), i}$ for some $i$, then deviating to price $v_{\pi_i(1), i}$ is profitable to seller $i$ since it will still sell to all the buyers and hence seller $i$ will earn strictly more profit at this new price.

Thus, $\textbf{p}$ is the unique market clearing pure Nash Equilibrium.
\end{proof}

We now show that theorem \ref{Theorem:twobuyersunique} extends to the case of m buyers when at most one of them is submodular.

\begin{theorem}
\label{Theorem:uniqueNashmbuyers}
When $N = [n]$, $M = [m]$, the valuation of one of the buyer is submodular and the valuations for the remaining buyers are additive, and $\forall i \in N$ $\forall j \in [m]$ $v_{\pi_i(j), i} - c_i < \frac{m}{m-j+1}(v_{\pi_i(1), i} - c_i)$ and $\forall i \in N$ $v_{\pi_i(1), i} > c_i$, then there is only one pure Nash Equilibrium and that equilibrium is market clearing. This equilibrium is given by $p_i = v_{\pi_i(1), i}$ $\forall i \in N$.
\end{theorem}

\begin{proof}
Let $\textbf{p}$ be an arbitrary pure Nash equilibrium. For each $i \in N$, $p_i \geq v_{\pi_i(1), i}$ since if $p_i$ is smaller than $v_{\pi_i(1), i}$, seller $i$ can increase the price and still sell to all the buyers as proved earlier in observation \ref{Observation:lp}. This will only increase his profit.

Let there be an item $i \in N$ that is not purchased by the submodular buyer at pricing $\textbf{p}$. Let $\pi_i(j)$ be the submodular buyer for some $j \in [m]$. Then, $p_i = v_{\pi_i(k), i}$ for some $k > j$ and item $i$ is purchased by buyers $\{\pi_i(l) | k \leq l \leq m\}$ (a total of m - k + 1) buyers. Thus, the profit of seller $i$ is $(m - k + 1)(v_{\pi_i(k), i} - c_i) < m(v_{\pi_i(1), i} - c_i)$ which will be the profit if the seller $i$ sets a price of $v_{\pi_i(1), i}$. Thus, $\textbf{p}$ is not an equilibrium. This is a contradiction and hence every item is purchased by the submodular buyer in an equilibrium. Therefore, in an equilibrium, the price of item $i$ is at most $v_{\pi_i(j), i}$ where $\pi_i(j)$ is the submodular buyer.

Let us assume that $p_i \neq v_{\pi_i(k), i}$ $\forall k \in [m]$ for some $i \in N$. Then, $v_{\pi_i(k), i} < p_i < v_{\pi_i(k+1), i}$ for some $k \in [m-1]$ since $v_{\pi_i(1), i} \leq p_i \leq v_{\pi_i(m), i}$. Now if the seller $i$ increases the price to $v_{\pi_i(k+1), i}$, the set of buyers that purchase item $i$ does not change and hence the profit of $i$ increases. Therefore, this is not possible in an equilibrium. Thus, $p_i = v_{\pi_i(k), i}$ for some $k \in [m]$ for all $i \in N$.

Let $p_i = v_{\pi_i(k), i}$ for some $k > 1$ for some $i \in N$. Then, the buyers $\{\pi_i(l) | k \leq l \leq m\}$ purchases the item $i$ and hence the profit of seller $i$ is $(m-k+1)(v_{\pi_i(k), i} - c_i)$. The seller can increase his profit by setting a price of $v_{\pi_i(1), i}$ since $m(v_{\pi_i(1), i} - c_i) > (m-k+1)(v_{\pi_i(k), i} - c_i)$. Thus, this is not possible in an equilibrium.

Therefore, $p_i = v_{\pi_i(1), i}$ $\forall i \in N$ gives the unique pure Nash equilibrium. This equilibrium is market clearing since all the buyers purchases all the items.
\end{proof}

The following corollary follows from the above theorem.

\begin{corollary}
\label{Corollary:mbuyersuniquePoA}
When $N = [n]$, $M = [m]$, the valuation of one of the buyer is submodular and the valuations for the remaining buyers are additive, and $\forall i \in N$ $\forall j \in [m]$ $v_{\pi_i(j), i} - c_i < \frac{m}{m-j+1}(v_{\pi_i(1), i} - c_i)$ and $\forall i \in N$ $v_{\pi_i(1), i} > c_i$, the price of anarchy is 1.
\end{corollary}


Next we give an example in case of $m$ buyers in which we show that the price of anarchy can be as high as $H_m$.

\begin{example}
\label{Example:welfbound}
Consider the market where $N = \{1\}$, $M = [m]$, $v_{i, 1} = \frac{1}{i}$ $\forall i \in [m] \setminus \{1\}$, $v_{1, 1} = 1 + \epsilon$ for some $\epsilon \geq 0$ and $c_1 = 0$. In this case, there is an equilibrium in which the only seller sets a price of $1 + \epsilon$ and at this price, only buyer $1$ purchases. The welfare of this equilibrium is $1 + \epsilon$. The optimal welfare, however, is $H_m + \epsilon$ (where $H_m$ is $m^{\textit{th}}$ harmonic number, $\sum_{i \in [m]} \frac{1}{i}$ which is approximately $\log{m}$) which is achieved when all the buyers purchases the item (at a price at most $\frac{1}{m}$). Therefore, for this equilibrium, the welfare is at most $\frac{1 + \epsilon}{H_m + \epsilon}$ of the optimal. 

\end{example}
Note that the example above can be extended to the case of $n$ sellers by taking additive valuations $v_j(S) = \frac{|S|}{j}$ for all $j \neq 1$ and $v_1(S) = |S|(1 + \epsilon)$. In the case $\epsilon > 0$ in the example above, there is only one equilibrium and the price of anarchy is at least $\frac{H_m + \epsilon}{1 + \epsilon}$ which is arbitrarily close to $H_m$ if $\epsilon$ is small enough. In the case that $\epsilon = 0$ and the sellers decide to offer the lowest price that maximizes their profit, the welfare is $H_m$ which is equal to the optimal welfare, and the price of anarchy in case of this tie breaking rule for sellers is 1.

The following theorem proves that in the case of no production costs, whenever an equilibrium exists, it has a welfare of at least $\frac{1}{H_m}$ of the optimal welfare and hence the price of anarchy is bounded by $H_m$. This bound is tight as shown in the example \ref{Example:welfbound} where there is an equilibrium with welfare exactly $\frac{1}{H_m}$ of the optimal welfare.

\begin{theorem}
\label{Theorem:PoAbound}
Suppose the set of sellers is $N = [n]$, set of buyers is $M = [m]$ and each seller $i$ has $m$ copies of a unique item also denoted $i$ with no cost of production. Suppose $\textbf{p}$ is a pure Nash equilibrium pricing vector under which each buyer $j$ buys the set $S_j$. Then,
\begin{align*}
\sum_{j \in M} v_j(N) \leq H_m \sum_{j \in M} v_j(S_j)
\end{align*}
\end{theorem}

We use the fact $v(S) + \sum_{i \in T \setminus S} v(i | T \setminus \{i\}) \leq v(T) \leq v(S) + \sum_{i \in T \setminus S} v(i | S)$ for $v$ submodular and $S \subseteq T$ quiet often in the following proof. This fact follows from the properties of submodularity.

\begin{proof}
Let $U = \bigcup_{j=1}^{m} S_j$ and $j \in M$ be fixed. Note that for all $i \in N \setminus U$, we have $v_j(i | S_j) = 0$ (otherwise, seller $i$ could have made non-zero profit by setting $p_i = v_j (i | S_j)$ due to observation \ref{Observation:prefbuy}). Hence, $v_j(N) \leq v_j(U) + \sum_{i \in N \setminus U} v_j(i | U) \leq v_j(U) + \sum_{i \in N \setminus U} v_j(i | S_j) = v_j(U)$. Now, monotonicity implies that $v_j(N) = v_j(U)$. So, if we show that $\sum_{j \in M} v_j(U) \leq H_m \sum_{j \in M} v_j(S_j)$, we are done.

Given $i \in N$ and $j \in M$, define
\begin{align*}
\textit{buy}(i) &= |j \in M \colon i \in S_j|,\\
\textit{rank}(j, i) &= |j' \in M \colon v_{j'}(i | S_{j'} \setminus \{i\}) \geq v_j(S_j \setminus \{i\})|.
\end{align*}

Note that buy(i) is the number of buyers who are currently buying from seller i, and rank(j, i) is the number of buyers who would buy from seller i if his price was equal to $v_j(i | S_j \setminus i)$. Also, note that the current profit of buyer i is $\textit{buy}(i)p_i$.

Note that for all $i \in N$ and $j \in M$ such that $i \notin S_j$, we must have $\text{rank}(j, i)v_j(i|S_j) \leq \text{buy}(i)p_i$
(otherwise seller $i$ could profit more by setting her price $v_j(i|S_j)$). Hence, $v_j(i|S_j) \leq \frac{\text{buy}(i)p_i}{\text{rank}(j, i)}$.
Now, fix $j \in M$ and note that
\begin{align*}
v_j(U) &\leq v_j(S_j) + \sum_{i \in U \setminus S_j} v_j(i | S_j) \\
&\leq v_j(S_j) + \sum_{i \in U \setminus S_j} \frac{\text{buy}(i)p_i}{\text{rank}(j, i)}
\end{align*}

Taking sum over $j \in M$, we have
\begin{align*}
\sum_{j \in M} v_j(U) &\leq \sum_{j \in M} v_j(S_j) + \sum_{j \in M} \sum_{i \in U \setminus S_j} \frac{\text{buy}(i)p_i}{\text{rank}(j, i)} \\
&= \sum_{j \in M} v_j(S_j) + \sum_{i \in U} \sum_{j \in M \colon i \notin S_j} \frac{\text{buy}(i)p_i}{\text{rank}(j, i)} \\
&= \sum_{j \in M} v_j(S_j) + \sum_{i \in U} \text{buy}(i)p_i (H_m - H_{\text{buy}(i)}) \\
&\leq \sum_{j \in M} v_j(S_j) + (H_m - 1) \sum_{i \in U} \text{buy}(i)p_i, \\
\end{align*}
where the third transition (second equality) holds because the buyers who do not buy from seller i have ranks buy(i) + 1, . . . , m, and the final transition (second inequality) holds because $i \in U$ implies buy$(i) \geq 1$(and thus $H_{\text{buy}(i)} \geq 1$).

It remains to show that $\sum_{i \in U} buy(i)p_i \leq \sum_{j \in M} v_j(S_j)$. To see this, fix $j \in M$ and note that
\begin{align*}
v_j(S_j) \geq \sum_{i \in S_j} v_j(i | S_j \setminus \{i\}) \geq \sum_{i \in S_j} p_i.
\end{align*}
Summing over $j \in M$, we get that
\begin{align*}
\sum_{j \in M} v_j(S_j) \geq \sum_{j \in M} \sum_{i \in S_j} p_i = \sum_{i \in U} \text{buy}(i)p_i,
\end{align*}
as desired.
\end{proof}

\section{The case of a single copy of each item and multiple buyers}
\label{sec:limited-supply}

In contrast to the results of Babaioff et al \cite{Babaioff2014}
for a single buyer, and our results in sections \ref{sec:multiple-copies} and \ref{sec:many-buyers}, the question of equilibria becomes more nuanced for multiple buyers with submodular valuations when each seller has a single copy of their item. We start by showing that the problem of pricing in the case of a single copy with costs to sellers reduces to the problem of pricing in the case of single copy with no costs to sellers and having one extra additive buyer. Conversely, if one of the buyer is additive and there are no costs to sellers, this can be reduced to the problem of one less buyer and costs to sellers.

Let $N = [n]$ and $M = [m]$ be the set of sellers and buyers respectively, with $c_i$ the cost of production to seller $i$. We also assume that each seller only produces one copy of their item. Let this problem be denoted by $P = (N, M, \textbf{c}, \{v_j\}_{j=1}^m)$ where $v_j$ is the valuation function of buyer $j$. We construct an equivalent problem $P' = (N', M', \{v'_j\}_{j=1}^{m+1})$ where $v'_j = v_j$ for $j \leq m$ is the valuation of buyer $j$, $N' = N$, $M' = [m+1]$ and $v'_{m+1}$ is an additive valuation of buyer $m+1$ with $v_{m+1, i} = c_i$ $\forall i \in [n]$. In $P'$ we assume that at any given pricing, an item is sold to buyer $m+1$ only if no one else purchases it. Note that there is no cost of production in $P'$. The following observation proves that $P$ and $P'$ are equivalent.

\begin{lemma}
\label{Lemma:cs}
Assuming that item $i$ not being sold implies $p_i = c_i$, there is a one-to-one correspondence in the equilibria in $P$ and $P'$.
\end{lemma}

\begin{proof}
If profit to seller $i$ under $\textbf{p}$ in $P$ is $0$, then it will earn exactly $c_i$ profit in $P'$ by selling to buyer $m+1$. If profit to seller $i$ under $\textbf{p}$ in $P$ is not $0$, the same buyer will still purchase in $P'$ but there will be no cost of production, so the profit will be $c_i$ more. By selling to someone else, the profit will not change. If some seller can earn more profit by increasing its price in $P$, the same seller can earn more profit in $P'$ since the difference in profit is constant at the same price; hence there is a one-to-one correspondence in the equilibria between $P$ and $P'$.
\end{proof}

Note that we have not made any assumptions on the valuation functions of the buyers or their budgets. Given lemma \ref{Lemma:cs} and the fact that $P'$ has no costs of production, with the exception of Corollary \ref{Corollary:costsingle}, we will assume in this section that costs to sellers is zero.

%
%

Without loss of generality, we assume that for each $i \in N$, $\exists j \in M$ such that $v_{j, i} > 0$. 
The first observation that we make is that all equilibria are market clearing. 
For this, we do not assume any restrictions on the buyers other than that all buyers are submodular and there are no costs to sellers. There may or may not be more than one submodular buyers.

\begin{observation}
\label{Observation:1}
If $\textbf{p}$ is an $\epsilon$-Nash Equilibrium (for $\epsilon < \min_{i \in N} \{\max_{j \in M} v_{j, i}\}$) in the case of $m$ buyers when their valuation functions are submodular and each seller has a single copy of their unique item, then $\textbf{p}$ is market clearing. It also follows that every pure Nash equililibrium is market clearing since it is a special case with $\epsilon = 0$.
\end{observation}

\begin{proof}
Let us assume that at pricing $\textbf{p}$, the item of seller $i$ is not being sold. Then, $p_i > \max_{j \in M} \{v_{j, i}\} > 0$. By setting a price of $\max_{j \in M} \{v_{j, i}\}$, the seller will sell his item at a price more than $\epsilon$ and hence this deviation is profitable. So, $\textbf{p}$ is not an $\epsilon$-Nash Equilibrium. This is a contradiction, and hence the item of seller $i$ is sold. Since $i$ is arbitrary, $\textbf{p}$ is market clearing.
\end{proof}

The assumption of submodularity is required in the previous observation as illustrated in the following example. Without this assumption, there is no guarantee that the market will be cleared.

\begin{example}
\label{Example:nosubnoclear}
Let us assume that there is a single buyer with valuation $v(1) = v(2) = 12, v(3) = 19, v(\{1, 2\}) = v(\{1, 3\}) = v(\{2, 3\}) = 20$ and $v(\{1, 2, 3\}) = 30$. It is easily verified that $v$ is subadditive. Let us assume that sellers set price $p_1 = p_2 = p_3 = 9$. In this case, the buyer will purchase $\{3\}$ even though $p_i < v(i|\{1,2,3\} \setminus \{i\})$ $\forall 1\leq i \leq 3$. Note that this pricing is not a pure Nash equilibrium. This shows that a price less than marginal value is not sufficient for the clearing of market if the buyers are not submodular.
\end{example}

\subsection{One submodular buyer with multiple additive buyers}
\label{subsec:one-submodular-buyer}

We will now consider a special case in which there is only one buyer who has an arbitrary submodular valuation function while the other buyers have additive valuations. 

First we note that in the particular case when all the $v_j$ are additive, there is an equilibrium where the seller $i$ sets a price of $\max_{j \in M} \{v_{j, i}\}$. This follows since the seller only has one copy of his item to sell, so the seller will sell it to the one who is willing to pay the highest price.

Now, we consider the case of $m$ buyers when $v_1$ is submodular and $v_j$ is additive for $2 \leq j \leq m$. This problem has an $\epsilon$-Nash Equilibrium which is proved in the following theorem. The proof is similar to the one in \cite{Babaioff2014} in which they prove the existence of $\epsilon$-Nash Equilibrium in the case of one submodular buyer and costs to sellers. Similar to them, we assume that the map that chooses the set purchased by buyer 1 is up-consistent. An up-consistent map is one in which if price of an item $i$ is increased, either the same set is chosen or item $i$ is not chosen. One of the examples of an up-consistent map is a map that chooses lexicographically the first set (the set that comes first in the alphabetical order if the set of alphabets is N and $i$ comes before $j$ if and only if $i < j$) as shown in \cite{Babaioff2014}. Gross-substitutes is a stronger condition than up-consistency since in the former, the items whose price is not changed are sold irrespective of whether the item whose price is increased is sold or not. Also, gross-substitutes condition holds when pricing of multiple items is increased at once which might not be the case for an up-consistent map.

\begin{theorem}
\label{Theorem:epsN}
When $N = [n]$, $M = [m]$, the first buyer is submodular and the remaining buyers are additive, then for any $\epsilon > 0$ there is a market clearing $\epsilon$-Nash equilibrium in the game of pricing between sellers when the map of buyer 1 is up-consistent. We assume that, at any given price, sellers prefer to sell their item to $j$ over $k$ if $j < k$.
\end{theorem}

\begin{proof}
We begin by setting prices $p_i = \max_{j \in M \setminus \{1\}} v_{j, i}$ $\forall i \in N$. We find a set $S$ under the pricing $\textbf{ip} = \textbf{p} = \{p_i\}_{i=1}^n$ that is purchased by buyer 1. While there is a seller $i$ that can increase the price to $p'_i$ which is at least $p_i + \epsilon$ such that at this new pricing $\textbf{p}' = (p'_i, p_{-i})$, the item of seller $i$ is sold, we update the price to $p_i + \epsilon$ and proceed. This process terminates since item $i \in S$ is sold at a maximum price of $v_1(i)$.

Initially at price $\textbf{ip}$, all the items are sold. Therefore, the profit of seller $i$ is $\max_{j \in M \setminus \{1\}} v_{j, i}$ $\forall i \in N$. Now, if some seller's price is updated, their profit increases since they increase the price only if they are still able to sell to buyer $1$. And if price of seller $i$ is increased, every other seller still has their item sold since the map of buyer 1 is up-consistent. Let this final pricing be $\textbf{fp}$. We claim that $\textbf{fp}$ is an $\epsilon$-Nash equilibrium.

At price $\textbf{fp}$, all the items are sold since they were sold in the beginning and increasing the price does not change the fact due to up-consistency. Therefore, no seller has an incentive to lower the price. If there is some seller $i$ that sells to buyer $1$ that can increase his profit by more than $\epsilon$ by increasing the price of his item, the price would have increased by at least $\epsilon$ and item $i$ would still have been sold contradicting the definition of $\textbf{fp}$. If there is some seller $i$ that does not sell to buyer $1$, buyer 1 would not purchase even if the seller increases the price of his item, and no other buyer will purchase at a price more than his current price $\max_{j \in M \setminus \{1\}} v_{j, i}$ and thus it is not beneficial for the seller to increase the price. Therefore, $\textbf{fp}$ is an $\epsilon$-Nash equilibrium. Since all the items are sold, this equilibrium is market clearing.
\end{proof}

Example \ref{Example:4.1} does not obey Theorem \ref{Theorem:epsN} because in contrast to the Theorem, the additive buyer is preferred over the submodular buyer by the sellers in the example.

Now, we give a corollary for the case when sellers have costs of production. The following corollary follows immediately from Theorem \ref{Theorem:epsN} and Lemma \ref{Lemma:cs}.

\begin{corollary}
\label{Corollary:costsingle}
When $N = [n]$, $M = [m]$, the first buyer is submodular and the remaining buyers are additive, and there are possible costs of production to sellers, then there is an $\epsilon$-Nash equilibrium in the game of pricing between sellers when the map of buyer 1 is up-consistent. We assume that, at any given price, sellers prefer to sell their item to $j$ over $k$ if $j < k$.
\end{corollary}

Note that in the case of production costs, the equilibria might not be market clearing if the costs are sufficiently high.

Now we give two examples that do not follow from Theorem \ref{Theorem:epsN} since the theorem does not say anything about pure Nash equilibria. We give an example where pure Nash equilibria exists and an example where it does not.

Now, we will give an example of a market and a pure Nash Equilibrium in it.

\begin{example}
\label{Example:3}
Consider the market where $M = [2]$, $N = [3]$ and the valuation function is given by the following table:
\begin{center}
\begin{tabular}{ c| c c c c c c c }
 sets & 1 & 2 & 3 & 1, 2 & 1, 3 & 2, 3 & 1, 2, 3 \\
 \hline
 $v_1$ & 110 & 112 & 114 & 123 & 125 & 127 & 136 \\
 $v_2$ & 10 & 12 & 14 & 22 & 24 & 26 & 36 \\
\end{tabular}
\end{center}
$v_1$ is submodular and $v_2$ is additive as can be easily verified. $\textbf{p} = (10, 12, 14)$ is the unique pure Nash Equilibrium and this equilibrium is market clearing when buyer $1$ is preferred by every seller over buyer $2$ at any given pricing if both the buyers are willing to purchase. We see that $\forall i$, $p_i = v_{2, i} > v_{1, i}$ gives the unique pure Nash equilibrium.
\end{example}

\begin{proof}
Suppose there is a Nash Equilibrium $\textbf{p}$ and under pricing $\textbf{p}$, $S_1$ is sold to $1$ and $N \setminus S_1$ to $2$.

If $S_1 = \emptyset$, then $p_i = v_{2, i}$ $\forall i \in N$ and at this price, buyer $1$ will prefer $\{1\}$ over $\emptyset$ and hence this case is not possible.

If $S_1 = \{1\}$, then $p_2 = 12$ and $p_3 = 14$ since they are sold to buyer $2$. Also, $p_1 = 10$ since if $p_1 > 10$, buyer $1$ will prefer to buy $2$ over $1$ and if $p_1 < 10$, $1$ can increase its price to $10$ and still be sold. In this case, buyer $1$ will buy $\{1, 2\}$ over $1$ since it gives him more utility. So, this case is not possible. Similarly, the cases $S_1 = \{2\}$ and $S_1 = \{3\}$ are not possible.

If $S_1 = \{1, 2\}$, then $p_3 = 14$ since it is sold to buyer $2$. Also, $p_1 \geq 10$ and $p_2 \geq 12$ since otherwise they can increase their price and still be sold to buyer $2$. The utility of buyer $1$ from $\{1, 2\}$ is $123 - p_1 - p_2$ and from $\{2, 3\}$ is $127 - p_3 - p_2 = 113 - p_2$, so if $p_1 > 10$, buyer $1$ will prefer $\{2, 3\}$ over $\{1, 2\}$. So, $p_1 = 10$. Similarly, $p_2 = 12$. The same thing happens for the other cases when $|S_1| = 2$.

If $S_1 = N$, then $p_i \geq v_{2, i}$ $\forall i$ since otherwise some seller can increase price to get more profit. At this price, utility from $S_1$ is $136-p_1-p_2-p_3 = 123 - p_1 - p_2 + 13 - p_3 < 123 - p_1 - p_2$ and hence buyer $1$ will prefer $\{1, 2\}$ over $N$, so this case is not possible.
\end{proof}

In the following example, there is no pure Nash equilibrium. A similar example for the case of costs to sellers was discussed in \cite{Babaioff2014}. In this example, we see the sufficient condition for multiple copies case ($\forall i \in N$ $\frac{v_{1, i}}{2} \leq v_{2, i} \leq 2v_{1, i}$) is satisfied but that does not guarantee the existence of pure Nash Equilibrium in the case of a single copy which exists in the case of two copies.

\begin{example}
\label{Example:4}
Consider the market where $M = [2]$, $N = [2]$ and the valuation function is given by the following table:
\begin{center}
\begin{tabular}{ c| c c c }
 sets & 1 & 2 & 1, 2 \\
 \hline
 $v_1$ & 14 & 14 & 25 \\
 $v_2$ & 10 & 12 & 22 \\
\end{tabular}
\end{center}
$v_1$ is submodular and $v_2$ is additive as can be easily verified. If buyer 1 prefers to buy item 2 over item 1, there is no pure Nash equilibrium when buyer $1$ is preferred by every seller over buyer $2$. Note that the map of buyer $1$ is up-consistent.
\end{example}

\begin{proof}
Let $\textbf{p}$ be a pricing with $\min \{p_1, p_2\} > 12$. In this case, buyer 2 will not purchase anything. Also, buyer 1 will not purchase both the items since the utility from purchasing item 1 is $14-p_1$ and utility from purchasing both the items is $25 - p_1 - p_2 = 14 - p_1 + 11 - p_2 < 14 - p_1$. Thus, at least one of the sellers whose item is not sold, has an incentive to set price $v_{2, i}$ to receive a higher profit.

Let $\textbf{p}$ be a pricing with $p_1 < p_2$ and $p_1 \leq 12$. Then, seller $1$ can increase the price to some number between $p_1$ and $p_2$ which is less than 14 and still sell to buyer 1. Thus, the seller has an incentive to increase the price.

Let $\textbf{p}$ be a pricing with $p_2 < p_1$ and $p_2 \leq 12$. Then, seller 2 can increase the price to some number between $p_2$ and $p_1$ which is less than 14 and still sell to buyer 1. Thus, the seller has an incentive to increase the price.

Let $\textbf{p}$ be a pricing with $p_1 = p_2 < 12$. Then, seller 2 can set price to 12 and still be sold to buyer 2, thus increasing profit.

Let $\textbf{p}$ be the pricing with $p_1 = p_2 = 12$. Then, the buyer 1 will purchase only item $2$ and hence seller $1$ has an incentive to deviate to a price less than $12$ giving him positive profit.

Since in all cases, there is an incentive to someone to change their price, there is no pure Nash equilibrium.
\end{proof}

The previous example shows that if all sellers have agreed on the same ordering of buyers, then there might not be a pure Nash Equilibrium.
In the following example we provide a game in which each seller sets both their price and their preference for whom to sell. We show that this game also does not have a pure Nash equilibrium.

\begin{example}
\label{Example:noequichoice}
Consider the market where $M = [2]$, $N = [3]$ and the valuation function is given by the following table:
\begin{center}
\begin{tabular}{ c| c c c c c c c }
 sets & 1 & 2 & 3 & 1, 2 & 1, 3 & 2, 3 & 1, 2, 3 \\
 \hline
 $v_1$ & 16 & 16 & 16 & 30 & 30 & 30 & 43 \\
 $v_2$ & 17 & 17 & 17 & 31 & 31 & 31 & 43 \\
\end{tabular}
\end{center}
$v_1$ and $v_2$ are submodular as can be easily verified. We consider the game in which each seller sets a price along with his own preference over buyers. We prove that there is no pure Nash equilibrium in this game for the valuations considered.
\end{example}

\begin{proof}
Suppose there is a pure Nash equilibrium given by pricing $\textbf{p}$ and some preferences of sellers. $p_1, p_2, p_3 \geq 13$ since at a price of 13 for item $i$, the buyer 1 always purchases item $i$ since $v_{1, i} = 13$. Therefore, the item of each seller is sold at equilibrium.

Let us assume that $p_1 > 14$. Then, the buyer that purchases item 1 will not purchase any other item since $v_j(1|\{2\}) = v_j(1 | \{3\}) = 14$ for both $j = 1$ and $j = 2$. Without loss of generality, let buyer 2 purchases item 1. Then, the other buyer (buyer 1) will purchase item $2$ and $3$. Thus, the price of item 2 and 3 will be at most 14. Seller $2$ and $3$ prefers buyer 1 over buyer 2 since otherwise buyer 2 will purchase one of them instead of $1$. Now, seller 2 can deviate to a price between $14$ and $p_1$ and prefer buyer 2 over buyer 1 resulting in selling item to buyer 2 at an increased price resulting in more profit. Thus, there is no equilibrium where $p_1 > 14$. By symmetry, $p_1, p_2, p_3 \leq 14$.

Now, we consider various cases. The first case is when all sellers prefer buyer i over 3 - i. In this case, one of the sellers can deviate to a price of 16 and prefer to sell to buyer 3 - i. This will let him profit more.

The second case is when two sellers prefer buyer i and third seller prefer buyer 3 - i. Since all the prices are at most 14, the third seller can deviate to price of 16 and still sell to buyer 3 - i earning strictly more profit.

Irrespective of pricing and preferences, there is always a strictly profitable deviation to at least one of the seller and hence there is no pure Nash equilibrium.
\end{proof}



\section{Conclusion and Open problems}
We have begun an extension of the pricing game for sellers when there are multiple submodular buyers and sellers have costs to produce their item. When there is enough supply (i.e., at least as many copies of each seller's item as number of buyers), then we have a reasonably clear understanding of equilibria and market clearing equilibria. Once we have a limited supply, we are far from a good understanding of when we do and when we do not have pure Nash equilibria. It is clear that having multiple submodular buyers (with or without budgets) and multiple sellers with costs to produce their items, raises challenges that have not been present in previous market pricing analysis where there is a single buyer. 

At a very general level, there are two obvious questions to pursue. First, in the full information setting, what conditions make it possible to have well-defined characterizations for when pure equlibria (and market clearing equilibria) exist? And given that pure equilibria exist, what can one say about the price of anarchy and price of stability in more general settings. In particular, the limited supply setting creates the most challenges. 
A second basic question is whether there are natural and efficient ways to arrive at equilibria when sellers do not precisely know the buyer valuations and yet such equilibria are known to exist as in the unlimited item case when there are no budgets. A third basic question concerns any of our multi-buyer settings when we further introduce buyer budgets as studied in Borodin et al \cite{Borodin2016} for the single buyer setting. In Appendix \ref{sec:budget}, we provide an equlibrium result for a very restricted unlimited supply setting, namely where there is one additive buyer with a budget and one additive buyer without a budget. We do not know if there is always an equilibrium when all buyers have additive valuations and budgets. A fourth basic question is what results might hold for more general cost functions (e.g., monotone submodular cost functions).

Given that there are so many directions for future work, it is perhaps best to start with some specific questions where we are most likely to make progress. In that regard, 
one specific question relating to the limited item setting in section \ref{subsec:one-submodular-buyer} is whether there is a ``best order'' for buyer arrival in the sense that for every instance of submodular buyer valuations, there is an ordering of the buyers that will allow sellers to have (market clearing) equilibirum prices. That is, given this best ordering of arrival, buyers will take any available item that is priced at or below their valuation. 
Finally, we note that our tie-breaking rule for buyers is to always break ties in favor of maximum cardinality (and arbitrarily amongst sets of equal size). This implies that a buyer will purchase an item even if the price is equal to the marginal value of the item. This is precisely the idea of a ``selfish agent'' as studied in 
Azar et al \cite{AzarBFFS19} in the context of unit demand buyers (i.e. bipartite matching) and social welfare whereas our primary consideration concerns more general submodular valuations and seller profit. It would be interesting to study other tie-breaking rules such as breaking ties in favor of minimum size which would imply that a buyer will not purchase unless there is a positive increase in his utility. 

\vspace{.2in}

{\large {\bf Acknowledgement}}

We are indebted to Omer Lev and Nisarg Shah for many constructive comments and improvements. In particular, Nisarg Shah provided the proof for Theorem \ref{Theorem:PoAbound}.


\bibliography{pricing-arxiv}

\appendix

\section{A set of conditions  characterizing market clearing equilibrium in case of two additive buyers, one with and one without a budget}
\label{sec:budget}


As evidenced in the single buyer case studied in Borodin et al \cite{Borodin2016}, the introduction of budgets changes the nature of equilibria when we add budgets. In this appendix, we consider one limited scenario where we can characterize the set of Nash equilibria. Namely, 
we consider the case of two additive buyers, one of whom (buyer 1) is budgeted (with a budget $B$). We assume that there are two copies of each item and there are no costs of production to sellers.

We start by making an observation about when the item will not be sold to buyer 1.

\begin{observation}
\label{Observation:3}
If for some $i \in N$, $v_{2, i} > 2v_{1, i}$; then the item $i$ will not be sold to buyer 1 in any equilibrium.
\end{observation}

\begin{proof}
At any pricing $\textbf{p}$ at which item $i$ is sold to buyer 1, the price is at most $v_{1, i}$ and hence the profit of seller $i$ is at most $2v_{1, i}$. If seller $i$ deviates to a price of $v_{2, i}$, it will still sell to buyer 2 and hence earn a profit of $v_{2, i}$. Thus, no pricing at which item $i$ is sold to buyer 1 is a pure Nash Equilibrium since seller $i$ has an incentive to deviate.
\end{proof}

The case $v_{2, i} > 2v_{1, i}$ is therefore relatively uninteresting and hence we assume that $v_{2, i} \leq 2v_{1, i}$ for all $i$ for further discussion. Let us denote by $S_1$ the set of items purchased by buyer 1 and $S_2$ purchased by buyer 2.
We will now make a further series of observations leading us to a set of conditions that are necessary and sufficient for obtaining a market clearing Nash equilibirum. 

\begin{observation}
\label{Observation:4}
If $\textbf{p}$ is a pure Nash Equilibrium, $i \in S_1$ and $p_i \neq v_{2, i}$ for some $i \in N$, then $v_{1, i} - p_i = \min_{j \in S_1} (v_{1, j} - p_j)$.
\end{observation}

\begin{proof}
If $j \in S_1$, then $v_{1, j} - p_j \geq 0$ since otherwise buyer 1 has more utility not purchasing item $j$. Let us assume to the contrary that $v_{1, i} - p_i > \min_{j \in S_1} (v_{1, j} - p_j)$. Let $k \in S_1$ be such that $v_{1, k} - p_k = \min_{j \in S_1} (v_{1, j} - p_j)$. We consider two cases: either $p_i < v_{2, i}$ or $p_i > v_{2, i}$.

We first consider the case $p_i < v_{2, i}$. Then, if seller $i$ deviates to a price $p'_i = \min \{p_i + \frac{(v_{1, i} - p_i) - (v_{1, k} - p_k)}{2}, v_{2, i}, B - \sum_{j \in S_1 \setminus \{i, k\}} p_j\}$, at this new pricing $\textbf{p}'$, item $i$ will be purchased by buyer 2. It will also be purchased by buyer 1 since $v_{1, i} - p'_i \geq \frac{(v_{1, i} - p_i) + (v_{1, k} - p_k)}{2} > v_{1, k} - p'_k$ and $\sum_{j \in S_1 \setminus \{k\}} p'_j \leq \sum_{j \in S_1 \setminus \{k, i\}} p_j + B - \sum_{j \in S_1 \setminus \{i, k\}} p_j \leq B$. Thus, the profit of seller $i$ is increased.

Next we consider the case $p_i > v_{2, i}$. Then, if seller $i$ deviates to a price $p'_i = \min \{p_i + \frac{(v_{1, i} - p_i) - (v_{1, k} - p_k)}{2}, B - \sum_{j \in S_1 \setminus \{i, k\}} p_j\}$, at the new pricing $\textbf{p}'$, item $i$ will be purchased by buyer 1 since $v_{1, i} - p'_i \geq \frac{(v_{1, i} - p_i) + (v_{1, k} - p_k)}{2} > v_{1, k} - p'_k$ and $\sum_{j \in S_1 \setminus k} p'_j \leq \sum_{j \in S_1 \setminus \{k, i\}} p_j + B - \sum_{j \in S_1 \setminus \{i, k\}} p_j \leq B$. Thus, the profit of seller $i$ is increased.

Thus, there is an incentive to deviate in both case and hence $\textbf{p}$ is not an Equilibrium. This is a contradiction and hence $v_{1, i} - p_i = \min_{j \in S_1} (v_{1, j} - p_j)$.
\end{proof}

From the above observation, it is clear that in any equilibrium either the price of an item $i$ will be equal to $v_{2, i}$ or it will be equal to $v_{1, i} - \min_{j \in S_1} (v_{1, j} - p_j)$. Let us investigate market clearing pure Nash Equilibrium. One of the necessary conditions for the existence of market clearing pure Nash Equilibrium is $v_{2, i} \leq 2v_{1, i}$ which follows from observation \ref{Observation:3}. Let us assume that $i < j \implies v_{1, i} - v_{2, i} \geq v_{1, j} - v_{2, j}$ $\forall i, j \in N$ without loss of generality. Then, we have the following observation.

\begin{observation}
\label{Observation:5}
If $\textbf{p}$ is a market clearing pure Nash Equilibrium such that for some item $i \in N$, $p_i = v_{2, i}$ then $p_j = v_{2, j}$ $\forall j < i$.
\end{observation}

\begin{proof}
Let us assume that for some $j < i$, $p_j \neq v_{2, j}$. Then, we have $v_{1, j} - p_j = \min_{k \in N} (v_{1, k} - p_k)$ from observation \ref{Observation:4}. Therefore, $v_{1, j} - p_j \leq v_{1, i} - p_i = v_{1, i} - v_{2, i}$ and since $j < i$, therefore $v_{1, j} - v_{2, j} \geq v_{1, i} - v_{2, i}$ and hence $v_{1, j} - p_j \leq v_{1, j} - v_{2, j}$. Thus, $p_j \geq v_{2, j}$. Since $p_j \neq v_{2, j}$, $p_j > v_{2, j}$ and hence item $j$ will not be purchased by buyer 2 under pricing $\textbf{p}$. This is a contradiction to the fact that $\textbf{p}$ is a market clearing pure Nash Equilibrium.
\end{proof}

By a similar argument as above, it follows that if $v_{1, i} - v_{2, i} = v_{1, j} - v_{2, j}$ and $p_i = v_{2, i}$ then $p_j = v_{2, j}$. Thus, there exists a $k \in N$ such that $\forall i \leq k$ $p_i = v_{2, i}$ and $\forall i > k$ $p_i = v_{1, i} - \min_{j \in N} (v_{1, j} - p_j)$. We will refer to this $k$ as {\it the boundary of price assignment}.

Using these observations, we can check for the existence of market clearing Nash Equilibria. Firstly, we check that $v_{2, i} \leq 2v_{1, i}$ $\forall i \in N$. This is a necessary condition as shown in observation \ref{Observation:3}. There can be two types of equilibrium: the one in which budget is consumed by buyer 1 and the one in which it is not.

We start with checking the case in which budget is not consumed. The boundary $k$ in this case is the largest index $i$ such that $v_{1, i} - v_{2, i} > 0$. We check the following conditions.

\begin{condit}
The following conditions are required for the case where budget is not consumed.
\begin{itemize}
\label{cond1}
\item $v_{2, i} \leq 2v_{1, i}$ $\forall i \in N$
\item $\sum_{i \in N} p_i < B$ (to ensure that we are in the non-budget consuming case)
\item $p_i = \min \{v_{1, i}, v_{2, i}\}$ $\forall i \in N$
\item For all $i$, $A \subseteq [k] \setminus \{i\}$ such that $p_i = v_{2, i}$, $v_{1, i} > 2v_{2, i}$ and $2v_{2, i} + \sum_{j \in [k] \setminus (A \cup \{i\})} v_{2, j} < B$ it must be that $\sum_{j \in A} (v_{1, j} - v_{2, j}) \geq v_{1, i} - 2v_{2, i}$ (we are ensuring that if seller $i$ increases the price to more than $2v_{2, i}$, buyer 1 will not gain more utility by still purchasing item $i$ and not purchasing some other items)
\end{itemize}
\end{condit}

\begin{lemma}
\label{lemma:necbudnot}
For a pricing $\textbf{p}$ with $\sum_{i \in N} p_i < B$, if $\textbf{p}$ is a market clearing pure Nash equilibrium then set of conditions \ref{cond1} is satisfied.
\end{lemma}

\begin{proof}
Let $\textbf{p}$ be a market clearing pure Nash equilibrium such that $\sum_{i \in N} p_i < B$.
If $p_i > v_{1, i}$ or $p_i > v_{2, i}$ for some $i$, one of the buyer will not purchase $i$ contradicting that $\textbf{p}$ is market clearing.
Thus, $p_i \leq \min\{v_{1, i}, v_{2, i}\}$ $\forall i \in N$.
If $p_i < \min \{v_{1, i}, v_{2, i}\}$ for some $i \in N$, then seller $i$ can set the price to $p'_i = \min \{v_{1, i}, v_{2, i}, B - \sum_{j \in N \setminus \{i\}}\}$ and his item will still be sold to both the buyers since $\sum_{j \in N} p'_j \leq B$ and $p'_j \leq \min \{v_{1, i}, v_{2, i}\}$ $\forall j \in N$ contradicting that $\textbf{p}$ is an equilibrium because there is a profitable deviation for seller $i$. Therefore, $p_i = \min \{v_{1, i}, v_{2, i}\}$ $\forall i \in N$.

Let there exists $i \in N$ and $A \subseteq [k] \setminus \{i\}$ such that $v_{1, i} > 2v_{2, i}$, $2v_{2, i} + \sum_{j \in [k] \setminus (A \cup \{i\})} v_{2, j} < B$ and $\sum_{j \in A} (v_{1, j} - v_{2, j}) < v_{1, i} - 2v_{2, i}$. Let the seller $i$ deviates to a price $p'_i$ strictly between $2v_{2, i}$ and $\min \{v_{1, i}, B - \sum_{j \in [k] \setminus (A \cup \{i\})} v_{2, j}, v_{1, i} - \sum_{j \in A} (v_{1, j} - v_{2, j})\}$. Let the buyer 1 purchases set $S$ not containing $i$ at this new pricing. Then, $\sum_{j \in (S \setminus A) \cup \{i\}} p'_j \leq \sum_{j \in [k] \setminus A} p'_j \leq B$ and $\sum_{j \in (S \setminus A) \cup \{i\}} v_{1, j} - p'_j \geq \sum_{j \in S} (v_{1, j} - p'_j) - \sum_{j \in A} (v_{1, j} - p'_j) + (v_{1, i} - p'_i) > \sum_{j \in S} (v_{1, j} - p'_j)$ and therefore buyer 1 will always purchase item $i$ at this new pricing. Therefore, it is profitable for seller $i$ to deviate and hence $\textbf{p}$ is not a pure Nash equilibrium. Hence, the last condition must be satisfied for $\textbf{p}$ to be a pure Nash equilibrium.
\end{proof}

\begin{lemma}
\label{lemma:sufbudnot}
For a pricing $\textbf{p}$ with $\sum_{i \in N} p_i < B$, if set of conditions \ref{cond1} is satisfied, then $\textbf{p}$ is a market clearing pure Nash equilibrium.
\end{lemma}

\begin{proof}
Since $p_i = \min \{v_{1, i}, v_{2, i}\}$ $\forall i \in N$ and $\sum_{i \in N} p_i < B$, all the copies of all the items is sold and hence no seller has an incentive to lower the price. Also, since $v_{2, i} \leq 2v_{1, i}$, a seller with price $v_{1, i}$ has no incentive to increase the price. Similarly, if $v_{1, i} \leq 2v_{2, i}$ for some seller $i$, the seller does not have an incentive to increase the price.

If $v_{1, i} > 2v_{2, i}$ for some seller $i$, the only possible increase in profit is when seller $i$ sets a price of $v_{1, i} \geq p'_i > 2v_{2, i}$. Let us assume that buyer 1 purchases a set $S$ containing $i$ at this new pricing. Let $A = [k] \setminus (S \cup \{i\})$. Then, $\sum_{j \in [k] \setminus \{i\}} p'_j < B$ and $\sum_{j \in [k] \setminus \{i\}} (v_{1, j} - p'_j) = \sum_{j \in A} (v_{1, j} - p'_j) + \sum_{j \in S \setminus \{i\}} (v_{1, j} - p'_j) > \sum_{j \in S} (v_{1, j} - p'_j)$ since $A$ satisfies the condition of the lemma. Therefore, buyer 1 will prefer $[k] \setminus \{i\}$ over $S$. Thus, this deviation is not profitable. Since no seller has an incentive to deviate, $\textbf{p}$ is a pure Nash equilibrium.
\end{proof}

We also check the case in which the  budget is consumed. In this case, the pricing is done using a boundary $k$ of price assignment. We iterate over all $k \in N \cup \{0\}$ and for each $k$, we create a new pricing $\textbf{p}_k$ with $p_{k, i} = v_{2, i}$ $\forall i \leq k$ and $p_{k, i} = \frac{(n-k)v_{1, i} - \sum_{j = k+1}^{n}v_{1, j} - \sum_{j = 1}^{k}v_{2, j} + B}{n-k}$ $\forall i > k$ (this pricing comes from the fact that $v_{1, i} - p_i = v_{1, j} - p_j$ $\forall i, j > k$). We check the following conditions for this pricing.

\begin{condit}
The following conditions are required for the case where budget is consumed. These conditions need to be satisfied for some $k \in [n] \cup \{0\}$.
\begin{itemize}
\label{cond2}
\item $\forall i \in N$ $v_{1, i} \geq p_{k, i} \geq 0$
\item $\forall i \in N$ $v_{2, i} \leq 2p_{k, i}$ (otherwise seller $i$ can set price to be $v_{2, i}$ and earn more profit by selling to buyer 2)
\item $v_{1, k+1} - v_{2, k+1} < v_{1, k+1} - p_{k, k+1} \leq v_{1, k} - v_{2, k}$
\item For all $i$, $A \subseteq [k] \setminus \{i\}$, $C \subseteq [n] \setminus [k]$ such that $p_{k, i} = v_{2, i}$, $v_{1, i} - 2v_{2, i} > v_{1, k+1} - p_{k, k+1}$, $2v_{2, i} + \sum_{j \in [k] \setminus (A \cup \{i\})} v_{2, j} + \sum_{j \in [n] \setminus ([k] \cup C)} p_{k, j} < B$ it must be that $\sum_{j \in A \cup C} v_{1, j} - p_{k, j} \geq v_{1, i} - 2v_{2, i}$.
\end{itemize}
\end{condit}

The proofs for the  following two lemmas  are  similar to lemmas \ref{lemma:necbudnot} and \ref{lemma:sufbudnot}.

\begin{lemma}
\label{lemma:necbud}
For a pricing $\textbf{p}$ with $\sum_{i \in N} p_i = B$, if $\textbf{p}$ is a market clearing pure Nash equilibrium then for some $k \in [n] \cup \{0\}$ $\textbf{p} = \textbf{p}_k$ and set of conditions \ref{cond2} is satisfied.
\end{lemma}

\begin{proof}
Let $\textbf{p}$ be a pure Nash equilibrium such that $\sum_{i \in N} p_i = B$.
Then following observations \ref{Observation:4} and \ref{Observation:5}, there exists $k$ such that $p_i = v_{2, i}$ $\forall i \in [k]$ and $v_{1, i} - p_i = v_{1, i + 1} p_{i+1}$ $\forall i \in [n-1] \setminus [k]$. Since $\sum_{i \in N} p_i = B$, we have $\forall i \in [n] \setminus [k]$ $\sum_{j \in [n] \setminus [k]} v_{1, i} - (n - k) (v_{1, i} - p_i) = B - \sum_{j \in [k]} v_{2, j}$ and hence $p_i = \frac{(n - k)v_{1, i} - \sum_{j=k+1}^{n} v_{1, j} - \sum_{j=1}^{k} v_{2, j} + B}{n - k} = p_{k, i}$ $\forall i \in [n] \setminus [k]$. Therefore, $p_i = p_{k, i}$ $\forall i \in [n]$.

$\forall i \in N$ $v_{1, i} \geq p_i$ holds since otherwise if $v_{1, i} < p_{k, i}$, buyer 1 will not purchase $i$ contradicting that $\textbf{p}$ is market clearing. Also, $\forall i \in N$ $p_i \geq 0$ since otherwise seller $i$ can deviate to 0 price and earn more profit. $\forall i \in N$ $v_{2, i} \leq 2p_i$ since otherwise some seller $i$ can deviate to a price of $v_{2, i}$ and sell just to buyer 2 earning more profit. The condition $v_{1, k+1} - v_{2, k+1} < v_{1, k+1} - p_{k+1}$ follows from the fact that $p_{k+1} < v_{2, k+1}$ since otherwise it contradicts the definition of boundary pricing. $v_{1, k+1} - p_{k+1} < v_{1, k} - p_k$ must hold because of observation \ref{Observation:4}.

Now, let us assume that there exists $i \in [k]$, $A \subseteq [k] \setminus \{i\}$, $C \subseteq [n] \setminus [k]$ such that $v_{1, i} - 2v_{2, i} > v_{1, k+1} - p_{k+1}$, $2v_{2, i} + \sum_{j \in [k] \setminus (A \cup \{i\})} v_{2, j} + \sum_{j \in [n] \setminus ([k] \cup C)} p_j < B$ and $\sum_{j \in A \cup C} (v_{1, j} - p_j) < v_{1, i} - 2v_{2, i}$. Let the seller $i$ deviates to a price $p'_i$ strictly between $2v_{2, i}$ and $\min \{v_{1, i} - v_{1, k+1} + p_{k+1}, B - \sum_{j \in [k] \setminus (A \cup \{i\})} v_{2, j} - \sum_{j \in [n] \setminus ([k] \cup C)} p_j, v_{1, i} - \sum_{j \in A \cup C} (v_{1, j} - p_j)\}$. Let the buyer 1 purchases set $S$ not containing $i$ at this new pricing. Then, $\sum_{j \in (S \setminus (A \cup C)) \cup \{i\}} p'_j \leq \sum_{j \in [k] \setminus A} p'_j + \sum_{j \in [n] \setminus ([k] \cup C)} p'_j \leq B$ and $\sum_{j \in (S \setminus (A \cup C)) \cup \{i\}} (v_{1, j} - p'_j) \geq \sum_{j \in S} (v_{1, j} - p'_j) - \sum_{j \in A \cup C} (v_{1, j} - p'_j) + (v_{1, i} - p'_i) > \sum_{j \in S} (v_{1, j} - p'_j)$ and therefore buyer 1 will always purchase item $i$ at this new pricing. Therefore, it is profitable for seller $i$ to deviate and hence $\textbf{p}$ is not a pure Nash equilibrium. Hence, the last condition must be satisfied for $\textbf{p}$ to be a pure Nash equilibrium.
\end{proof}

\begin{lemma}
\label{lemma:sufbud}
For a pricing $\textbf{p}$ with $\sum_{i \in N} p_i = B$, if $\textbf{p} = \textbf{p}_k$ for some $k \in [n] \cup \{0\}$ and set of conditions \ref{cond2} is satisfied for this $k$, then $\textbf{p}$ is a market clearing pure Nash equilibrium.
\end{lemma}

\begin{proof}
Since $0 \leq p_i = p_{k ,i} \leq v_{1, i}$ $\forall i \in N$ and $\sum_{i \in N} p_i = B$, buyer 1 purchases all the items. Since $p_{k + 1} < v_{2, k+1}$, $v_{1, i} - p_i = v_{1, k+1} - p_{k + 1}$ $\forall i > k$ and $v_{1, i} - v_{2, i} \leq v_{1, k+1} - v_{2, k+1}$ we have that $p_i < v_{2, i}$ $\forall i > k$. As $p_i \ v_{2, i}$ $\forall i \in [k]$, all the items are purchased by buyer 2. Since the market is cleared, no seller has an incentive to the lower the price. Also, since $v_{2, i} \leq 2p_i$, a seller with price $p_i < v_{2, i}$ has no incentive to increase the price since if the seller $i$ increases the price, item $i$ will provide the least utility to buyer 1 and hence buyer 1 will not purchase it. Similarly, if $v_{1, i} - 2v_{2, i} \leq v_{1, k+1} - p_{k + 1}$ for some seller $i \in [k]$, the seller does not have an incentive to increase the price since to profit more, the price must be made larger than $2v_{2, i}$ at which buyer 1 will not purchase item $i$.

If $v_{1, i} - 2v_{2, i} > v_{1, k+1} - p_{k + 1}$ for some seller $i$, the only possible increase in profit is when seller $i$ sets a price of $v_{1, i} \geq p'_i > 2v_{2, i}$. Let us assume that buyer 1 purchases a set $S$ containing $i$ at this new pricing. Let $A = [k] \setminus (S \cup \{i\})$ and $C = [n] \setminus ([k] \cup S)$. Then, $\sum_{j \in [n] \setminus \{i\}} p'_j < B$ and $\sum_{j \in [n] \setminus \{i\}} (v_{1, j} - p'_j) = \sum_{j \in A \cup C} (v_{1, j} - p'_j) + \sum_{j \in S \setminus \{i\}} (v_{1, j} - p'_j) > \sum_{j \in S} (v_{1, j} - p'_j)$ since $A, C$ satisfies the condition of the lemma. Therefore, buyer 1 will prefer $[n] \setminus \{i\}$ over $S$. Thus, this deviation is not profitable. Since no seller has an incentive to deviate, $\textbf{p}$ is a pure Nash equilibrium.
\end{proof}


The following theorem follows from lemmas \ref{lemma:necbudnot}, \ref{lemma:sufbudnot}, \ref{lemma:necbud} and \ref{lemma:sufbud}.
\begin{theorem}
\label{Theorem:budsection}
There is a market clearing pure Nash equilibrium if and only if either the set of conditions \ref{cond1} is satisfied or the set of conditions \ref{cond2} is satisfied for some $k \in [n] \cup \{0\}$.
\end{theorem}
Thus, checking all the conditions above is necessary and sufficient to check the existence of market clearing pure Nash Equilibrium. There are $n+2$ possible market clearing pure Nash Equilibrium ($n+1$ budget clearing, one for each value of $k$ and one non-budget clearing), and we can check for all of them using the above conditions. These are the only possible market clearing pure Nash Equilibria.


\end{document}